\documentclass[preprint,aap]{Nimsart}
\RequirePackage[OT1]{fontenc}
\RequirePackage{amsthm,amsmath,natbib}

\usepackage{amsthm,amsmath,amssymb,natbib}
\newcommand{\mc}{\mathcal}
\newcommand{\mb}{\mathbb}
\newcommand{\E}{\mathbb{E}}

\newcommand{\ti}{\tilde}

\newcommand{\D}{\mathrm{Dom}}
\newcommand{\I}{\textbf{1}_}
\newcommand{\eps}{\varepsilon}

\theoremstyle{definition}
\newtheorem{defn}{Definition}[section]
\newtheorem{rmk}[defn]{Remark}

\theoremstyle{plain}
\newtheorem{lem}[defn]{Lemma}
\newtheorem{thm}[defn]{Theorem}

\newtheorem{ass}[defn]{Assumption}

\numberwithin{equation}{section}

\DeclareMathOperator*\esssup{ess \,sup}

\begin{document}
\begin{frontmatter}

\title{Constrained NonSmooth Utility Maximization on the Positive Real Line}  
\runtitle{Nonsmooth Utility Maximization with Constraints}


\author{\fnms{Nicholas} 
\snm{Westray}\ead[label=e1]{westray@math.hu-berlin.de}
\thanksref{t1}}
\thankstext{t1}{This research was funded by an EPSRC DTA Grant}
\address{Nicholas Westray\\ Department of Mathematics\\ Humboldt Universit\"{a}t Berlin \\ Unter den Linden 6, 10099 Berlin \\ Germany\\
\printead{e1}}

\and
\author{\fnms{Harry} \snm{Zheng}\ead[label=e2]{h.zheng@imperial.ac.uk}}
\address{Harry Zheng\\ Department of Mathematics\\ 
Imperial College \\London  SW7 2AZ \\ UK \\ \printead{e2}}

\affiliation{Humboldt Universit\"{a}t and Imperial College}

\runauthor{Nicholas Westray and Harry Zheng}


\begin{abstract}
We maximize the expected utility of terminal wealth in an incomplete market where there are cone constraints on the investor's portfolio process
and the utility function  is not assumed
to be strictly concave or differentiable. 
We establish the existence of the optimal solutions to the primal and dual
problems and their dual relationship.
We simplify the present proofs in this area and extend 
the existing duality theory to the constrained nonsmooth setting. 
\end{abstract}
\end{frontmatter}
\section{Introduction}
Utility maximization is a classical theme in mathematical finance
and there is already a substantial body of literature devoted to the study of the problem in both complete and incomplete semimartingale models. We refer the reader to Karatzas and
\v{Z}itkovi\'{c} \cite{KZ03} as well as Kramkov and Schachermayer \cite{KS99} for an excellent overview of research to date.
The purpose of the present article is to extend the existing duality theory to the situation where there are cone constraints on the investor's portfolio as well as a utility function which is neither smooth nor strictly concave. To set  the context for this paper we first review previous work in the area which is of immediate interest.  

Cvitani\'{c}, Schachermayer and Wang \cite{CSW01} solve the utility maximization problem with a bounded random endowment. 
They prove that the usual duality relations hold but to achieve this it is necessary to enlarge the dual domain from $L^1(\mb{P})$ to $L^\infty(\mb{P})^*$, the topological dual of $L^\infty(\mb{P})$. In \cite{KZ03} these ideas are extended further to include intertemporal consumption. More recently Hugonnier and Kramkov \cite{HK04},  using some elegant techniques from convex analysis, generalize the results of \cite{CSW01} to the case of unbounded random endowment.

There has been some work on applying duality theorems in utility maximization.
Bellini and Frittelli \cite{BF02} as well as Biagini and Frittelli \cite{BF05}
show that a version of the Fenchel duality theorem can be combined
with a characterization of conjugate functionals in $L^\infty(\mb{P}) ^*$ due to Rockafellar \cite{R71} to efficiently establish the existence of a dual solution as well as the equality of the value functions. 

The introduction of constraints on the investor's portfolio in a general semimartingale
model is relatively new. Mnif and Pham \cite{MP01}
as well as Pham \cite{P02} provide a solution to the problem when the
underlying market is modelled by a continuous semimartingale with positive definite quadratic variation matrix.

A standard assumption in almost all papers is that the utility function be strictly concave and continuously differentiable. In the present article we are interested in advancing the general theory and so want to consider the situation where this may not be the case. Cvitani\'{c} \cite{C98} first
addresses this when considering a framework similar to  \cite{P02} but where the loss function is neither strictly concave nor differentiable. He derives solutions using subdifferential calculus together with convex analysis. The first nonsmooth utility maximization problem appears in Deelstra, Pham and Touzi \cite{DPT01} (see also Bouchard \cite{B02} as well as Bouchard, Touzi and Zeghal \cite{BTZ04}). Their solution uses the
quadratic inf convolution method which, whilst mathematically very satisfying, leads to lengthy and involved proofs.

This article contributes in several ways to the existing
literature. Firstly we incorporate the distinct features of \cite{BF02,CSW01,DPT01}
into a single model and extend the setting further by allowing for cone constraints on the portfolio. We use a technique due to Kramkov and Schachermayer \cite{KS03} (see also \cite{B02} as well as Westray and Zheng \cite{WZ08}) to prove directly the existence of a solution to the primal problem and thus remove the need for quadratic inf convolution, simplifying the existing proofs in this area. Secondly we apply a new result from Czichowsky, Westray and Zheng \cite{CWZ08}, in conjunction with the constrained optional
decomposition theorem of F\"{o}llmer and Kramkov \cite{FK97}, to show that the restrictive assumptions on the underlying asset in \cite{MP01,P02} are redundant. Finally we use a version of the Fenchel duality theorem due to Rockafellar \cite{R66}, different from that in \cite{BF02,BF05}, to give
a simple proof of the existence of a dual solution.
 
This paper is organised as follows. Section \ref{chp_Uhalfline_model} introduces the model formulation. Section \ref{Linfinity*}
discusses the dual problem and provides the essential results on constrained super replication. Section \ref{chp_Uhalfline_proof} contains the main result, Theorem \ref{chp_Uhalfline_mainthm},
together with its proof.


\section{Model Formulation}
\label{chp_Uhalfline_model}
The setup is the standard one in mathematical finance. There is a finite time horizon $T$ and a market consisting of one bond, assumed constant, and $d$ stocks, $S^1,\ldots,S^d$ modelled by  a $(0,\infty)^d$-valued, semimartingale on a filtered probability space 
$(\Omega,\mc{F},(\mc{F}_t)_{0\leq t\leq T},\mb{P})$, satisfying the usual conditions. We also assume, for simplicity, that the initial $\sigma$-field
$\mc{F}_0$ is trivial.  We write $X$ for the process $(X_t)_{0\leq t\leq T}$ and ``for
all $t$" implicitly meaning ``for all $t\in[0,T]$". For a predictable $S$-integrable
process, we use $H\cdot S$ to denote the stochastic integral with respect to $S$ and refer the reader to Jacod and Shiryaev \cite{JS03} and Protter \cite{Pr05} for further details. 

We want to define those investment strategies which are admissible. In
the current setting there are constraints, modelled by the set $\mc{K}\subset
\mb{R}^d$. 
\begin{ass}
\label{chp_Uhalfline_Kassump}
$\mc{K}$ is a closed convex cone such that there exist
$m\in\mb{N}$ and $k_1,\ldots,k_m$ in $\mc{K}$ with 
\begin{equation*}
\mc{K}=\left\{\sum_{i=1}^m \mu_i k_i:\mu_i\geq0\right\}.
\end{equation*}
\end{ass}

\noindent This assumption states that $\mc{K}$ is a polyhedral convex cone, see Rockafellar \cite{R70} Theorem 19.1. This class of sets contains some interesting examples, including no short selling of the first $m$ assets, $\mc{K}=\mb{R}^m_+\times\mb{R}^{d-m}$.

It is known that to prevent arbitrage we must exclude some trading strategies such as doubling. We define $\mc{H}$, the set of \textit{admissible} trading strategies, as follows.
\begin{align*}
\mc{H}:=\left\{H:H \text{ predictable and $S$-integrable}, H_t\in\mc{K} \, \text{ $\mb{P}$-a.s.
for all } t \right.
\\\left. \text{ and there exists } c\in\mb{R}_+ \text{ with } (H\cdot S)_t\geq
-c \text{ for all } t \right\}.
\end{align*}

Next we introduce the cone of random variables which can be dominated (super
replicated) by terminal wealths obtained from admissible strategies.
\begin{equation*}
\mc{R}:=\left\{R:R\leq(H\cdot S)_T \text{ for some } H\in\mc{H}\right\}.
\end{equation*}
Since $0\in\mc{K}$ it follows that $L^0_{-}(\mb{P})\subset\mc{R}$.
We define  
\begin{equation*}
\mc{C}:=\mc{R}\cap L^{\infty}(\mb{P}).
\end{equation*}
The set $\mc{C}$ contains all those random variables which are bounded and super-replicable.
Our agent has preferences modelled by a 
utility function $U$, increasing, concave and satisfying 
\begin{equation*}
\mathrm{int}\big(\D(U)\big)=\{x\in\mb{R}:x>0\},
\end{equation*}
where $\mathrm{int}\big(\D(U)\big)$ is the interior of the domain of $U$. To avoid any ambiguity we set $U(x)=-\infty$ for $x<0$. Observe that we do not insist that $U$ be strictly concave or differentiable. We  make the following assumption, the nonsmooth analogue of the Inada conditions.
\begin{ass}
\label{chp_Uhalfline_UProps} 
\begin{equation*}
\inf\bigcup_{x\in\mb{R}_+}\partial U(x) = 0,\quad \sup\bigcup_{x\in\mb{R}_+}\partial
U(x)=\infty,
\end{equation*}
\end{ass}
\noindent where $\partial U(x)$ denotes the subdifferential (subgradient) set of $U$ at $x$, i.e., 
$\partial U(x):=\{\xi\in \mb{R}:
U(y)\leq U(x)+\xi(y-x),\;\forall y\in \mb{R}\}$.

The agent starts with an initial capital $x$, may choose strategies from $\mc{H}$,
and aims to maximize the expected utility of terminal wealth subject to a random endowment $B\in L^{\infty}(\mb{P})$ with $b:=\|B\|_{L^\infty(\mb{P})}$. This leads to the following formulation
of the primal maximization problem.  
\begin{equation}
\label{primalproblem}
u(x):=\sup_{R\in\mc{R}_0(x)}\E\big[U(x+R-B)\big],
\end{equation}
where $\mc{R}_0(x)$ is the set containing all those $R\in\mc{R}$ for which the above expectation is well defined for a given $x$.
\begin{rmk}
\label{chp_Uhalfline_expfinite}
Let us expand slightly on this final point. The expectation in \eqref{primalproblem}
will be well defined if and only if both $\E\big[U(x+R-B)^+\big]$ and $\E\big[U(x+R-B)^-\big]$ are not equal to $+\infty$. Since $U(x)=-\infty$ for $x<0$ and a priori $\mc{R}$
may contain random variables which take negative values or
large positive values  with positive probability, it may happen that this
expectation fails to be well defined. We restrict to $\mc{R}_0(x)$
to ensure that this situation does not arise.
\end{rmk} 
We write $\ti{U}$ for the \textit{conjugate} (or \textit{dual}) of $U$ defined by 
\begin{equation*}
\tilde{U}(y)=\sup_{x\in\mb{R}_+}\big\{U(x)-xy\big\}.
\end{equation*} 
This is a convex and decreasing function with $\D(\tilde{U})\cap(-\infty,0)=\emptyset$.
From \cite{KS99} it is known that to guarantee the existence of an optimal solution we must impose a condition on the asymptotic elasticity 
of the utility function $U$. In \cite{DPT01} the authors show that, for
a nonsmooth utility function, these should be put on the dual function. Define
\begin{equation}
\label{defnAEofU}
\mathrm{AE}(\tilde{U}):=\limsup_{y\to 0}\sup_{q\in\partial \tilde{U}(y)}\frac{|q|y}{\tilde{U}(y)}.
\end{equation}
We shall need the following.
\begin{ass}
\label{chp_Uhalfline_AEofU}
$\mathrm{AE}(\tilde{U})<\infty$.
\end{ass}


\section{Dual Domain and Dual Characterization of $\mc{R}$} 
\label{Linfinity*}
With the primal problem formulated and the dual of the utility function 
introduced, we move to consider the domain of the dual problem. In
our setting there is a bounded random endowment  so we shall follow Biagini, Frittelli and Grasselli \cite{BFG08} as well as \cite{CSW01}
and formulate the dual domain as a subset of $L^\infty(\mb{P})^*$. 

We first provide an introduction to the relevant theory of the topological
dual of $L^\infty(\mb{P})$, for further details see Hewitt and Stromberg \cite{HS65} as well as Hewitt and Yosida \cite{YH52}. 
We write $ba(\mb{P})$ for the set of bounded, finitely additive measures, absolutely continuous with respect to $\mb{P}$ and $ba_+(\mb{P})$ for the nonnegative elements of $ba(\mb{P})$. We shall indirectly use the following important decomposition theorem.
\begin{thm}[Yosida and Hewitt \cite{YH52} Theorem 1.23]
\label{YHthm}
If $\nu\geq0$ is in $ba_+(\mb{P})$ then there exist unique $\nu_c\geq
0,\nu_f\geq 0$, both in $ba_+(\mb{P})$, such that $\nu_c$ is countably additive, $\nu_f$ purely finitely
additive and $\nu=\nu_c+\nu_f$. 
\end{thm}
One can develop a theory for the integration of bounded random variables with respect to finitely additive measures, see \cite{HS65} for details. It then follows that each $\nu\in ba(\mb{P})$ induces a linear functional $\psi_{\nu}:L^\infty(\mb{P})\to\mb{R}$ defined by 
\begin{equation*}
\psi_{\nu}(G)=\int_\Omega G d\nu.
\end{equation*}
It is shown in \cite{HS65} Theorem 20.35 that the mapping
which takes $\nu$ to $\psi_\nu$ is an isometric isomorphism between
$ba(\mb{P})$ and $L^\infty(\mb{P})^*$. We may thus identify the set of
bounded finitely additive measures with the topological dual of $L^\infty(\mb{P})$,
i.e., $L^\infty(\mb{P})^*\cong ba(\mb{P})$. Furthermore $L^1(\mb{P})$ is isomorphic to the subspace of $ba(\mb{P})$ containing all countably additive measures, i.e., 
\begin{equation*}
L^1(\mb{P})\cong \{\nu\in ba(\mb{P}): \nu_f=0\}
\end{equation*}
and so we may view $L^1(\mb{P})$ as a subspace of $ba(\mb{P})$. This leads
to the following useful expression for a random variable $G\in L^\infty(\mb{P})$
and a countably additive element $\nu\in ba_+(\mb{P})$,
\begin{equation}
\label{defnofpsinu}
\psi_{\nu}(G)=\int_\Omega G \frac{d\nu}{d\mb{P}}d\mb{P}=\E\left[\frac{d\nu}{d\mb{P}}G\right].
\end{equation}

With the necessary preliminaries covered let us introduce the dual domain,
\begin{equation*}
\mc{M}:=\{\nu\in ba(\mb{P}):\psi_{\nu}(G)\leq0 \text{ for all } G\in\mc{C}\}.
\end{equation*}
We first collect some important observations about $\mc{M}$. Since
$L^{\infty}_-(\mb{P})\subset \mc{C}$, $\mc{M}$ is a cone contained
in $ ba_+(\mb{P})$.
In fact $\mc{M}=(\mc{C})^0$, the polar of $\mc{C}$ with respect to the dual
system $\big(L^\infty(\mb{P}),ba(\mb{P})\big)$, see Heuser \cite{He82} for background on the theory of dual systems. Note that unlike \cite{CSW01} it is
not assumed that all the elements in $\mc{M}$ have norm 1.  
The set of countably additive elements of $\mc{M}$ is defined by
\begin{equation*}
\mc{M}^c:=\{\nu\in \mc{M}:\nu_f=0\},
\end{equation*}
which one could think of as ``$\mc{M}\cap L^1(\mb{P})$'', equivalently all
those measures in $\mc{M}$ which have a Radon-Nikodym derivative.

In the context of this article it is necessary to extend the definition of
$\psi_{\nu}$, when $\nu\in\mc{M}$, to elements $X\in L^0(\mb{P})$ which are $\mb{P}$-a.s. bounded below. Set 
\begin{equation*}
\psi_{\nu}(X):=\lim_{n\to\infty}\psi_{\nu}(X\wedge n).
\end{equation*}
In particular, using the monotone convergence theorem, for $\nu\in\mc{M}^c$
and $X\in L^0(\mb{P})$ bounded below
\begin{equation*}
\psi_{\nu}(X):=\lim_{n\to\infty}\E\left[\frac{d\nu}{d\mb{P}}(X\wedge n)\right]
=\E\left[\frac{d\nu}{d\mb{P}}X\right].
\end{equation*}
Thus relation (\ref{defnofpsinu}) continues to hold for the extension.
We shall need the following.
\begin{ass}
\label{chp_Uhalfline_FGE}
There exists  $\nu^1\in\mc{M}$  with 
\begin{equation*}
\E\left[\ti{U}\left(\frac{d\nu_c^1}{d\mb{P}}\right)\right]<\infty.
\end{equation*}
\end{ass}
\begin{rmk}
The above is to ensure that the dual problem is finite for all $x>0$. A well known consequence of Assumption \ref{chp_Uhalfline_AEofU}
is that 
\begin{equation}
\label{dualfinite}
\E\left[\ti{U}\left(r\frac{d\nu^1_c}{d\mb{P}}\right)\right]<\infty \text{
for all } r>0.
\end{equation}
We shall return to this result later. 
\end{rmk}
There is a subset of $\mc{M}$ which
plays a key role in defining those $x$ for which the primal problem is finite,
\begin{equation}
\label{Msup=M}
\mc{M}^{\mathrm{sup}}:=\left\{\nu\in\mc{M}^c:\nu>0 \text{ and } \nu(\Omega)=1\right\}.
\end{equation}
Using an identical proof to that of  \cite{BF05} Proposition 6 it is possible to show that when $\nu$ is countably additive and $\nu(\Omega)=1$, 
$\psi_\nu(G)\leq 0$ for all $G\in \mc{C}$
 if and only if $H\cdot S$ is a $\nu$-supermartingale for all $H\in\mc{H}$. 
Therefore we have
\begin{gather*}
\nonumber
\mc{M}^{\mathrm{sup}}=\{\nu\in ba_+(\mb{P}) : \nu \text{ is a probability measure equivalent
to $\mb{P}$ and }\\
\text{ $H\cdot S$ is a $\mb{\nu}$-supermartingale
for all $H\in\mc{H}$ } \}.
\end{gather*}

\begin{ass}
\label{ELMM}
$$\mc{M}^{\mathrm{sup}}\ne\emptyset.$$
\end{ass}

\begin{rmk}
Suppose that we have no constraints, so that $\mc{K}=\mb{R}^d$. Our situation is now identical to 
\cite{CSW01} modulo the smoothness of the utility function. Consider the set 
\begin{gather*}
\mc{M}^{\mathrm{loc}}:=\{\nu\in ba_+(\mb{P}) : \nu \text{ is a probability measure equivalent
to $\mb{P}$ and }\\
 \text{ $H\cdot S$ is a $\mb{\nu}$ local martingale
for all $H\in\mc{H}$ } \}.
\end{gather*}
In \cite{CSW01} the authors assume that $\mc{M}^{\mathrm{loc}}\neq\emptyset$. We want to compare this to our Assumption \ref{ELMM} when there are no constraints. If $S$ is locally bounded then it is known that both are equivalent. In the case where $S$ is not locally bounded then we have $\mc{M}^{\mathrm{loc}}\subset\mc{M}^{\mathrm{sup}}$ and this inclusion may be strict. Thus it appears that our assumption is slightly weaker, however both these assumptions imply that the subset of $\mc{M}^c$ consisting of equivalent measures is nonempty and from this point of view may be regarded as equivalent. In particular all the results in \cite{CSW01} would hold
under our Assumption \ref{ELMM}.  

\end{rmk}

We look for  a description of $\mc{R}$
in terms of a budget constraint inequality. 
To this end we appeal to the ideas of
 \cite{FK97}. Define the set 
\begin{equation*}
\mc{S}:=\{H\cdot S : H\in\mc{H}\},  
\end{equation*} 
and consider the following assumption,
\begin{ass}[\cite{FK97} Assumption 3.1 ]
\label{Sass}
If $(H^n\cdot S)_{n\in\mb{N}}$ is a sequence in $\mc{S}$ which
is uniformly bounded from
below and converges in the semimartingale topology to a process $X$, then
$X\in\mc{S}$.
\end{ass}
For details on the semimartingale topology we refer the reader to \'{E}mery \cite{Em79} and M\'{e}min \cite{M80}. In our setting and notation \cite{FK97} Theorem 4.1 reads as follows.
\begin{thm}[\cite{FK97} Theorem 4.1]
\label{FKthm4.1}
Suppose that $\mc{M}^{\mathrm{sup}}\neq\emptyset$ and $\mc{S}$ satisfies Assumption
\ref{Sass}. Then for a process $V$ locally bounded from below the following are
equivalent:
\begin{enumerate}
\item There exist $H^V\in\mc{H}$ and an increasing nonnegative optional process $C^V$ such that
\begin{equation*}
V=V_0+H^V\cdot S-C^V.
\end{equation*}
\item $V$ is a $\nu$-local supermartingale for all $\nu\in\mc{M}^{\mathrm{sup}}$.
\end{enumerate}
\end{thm}
Our aim is to prove a super replication result via Theorem \ref{FKthm4.1}. 
We know that $\mc{M}^{\mathrm{sup}}$ is 
 nonempty, thus the only outstanding issue is to verify that $\mc{S}$ satisfies Assumption \ref{Sass}.
By Assumption \ref{chp_Uhalfline_Kassump} $\mc{K}$
is a polyhedral cone and thus we may apply  \cite{CWZ08} Theorem  3.5 
to show that there exists an 
$H^0\in\mc{H}$ such that 
$X=H^0\cdot S$. Therefore 
$\mc{S}$ satisfies Assumption \ref{Sass}.
\begin{rmk}
In \cite{KZ03} the authors consider a situation similar to ours however they
do not make any assumptions on the cone $\mc{K}$. They implicitly assume
that $\mc{S}$ satisfies Assumption \ref{Sass}. This is in fact false as a counterexample in \cite{CWZ08}
shows. Our results can therefore be viewed as augmenting \cite{KZ03} and
show that Assumption \ref{chp_Uhalfline_Kassump} is not innocuous and that
one must place some restrictions on $\mc{K}$.
\end{rmk}
We now give a key dual characterization of $\mc{R}$.
\begin{lem}
\label{InKiffnuK}
Suppose $R^-\in L^{\infty}(\mb{P})$. Then $R\in\mc{R}$ if and only if 
$\psi_{\nu}(R)\leq 0$ for all $\nu\in\mc{M}^c$.
\end{lem} 
\begin{proof}
Suppose $R\in\mc{R}$ and $R^-\in L^\infty(\mb{P})$. Then $R\wedge n\in\mc{C}$
and we have by the definition of $\mc{M}$ 
\begin{equation*}
\psi_{\nu}(R\wedge n)\leq 0 \text{ for all } \nu\in\mc{M}^c.  
\end{equation*}
Now let $n$ tend to infinity to get the result. 

Conversely suppose that $R^-\in L^{\infty}(\mb{P})$ and $\psi_{\nu}(R)\leq 0$ for all
$\nu\in\mc{M}^c$. In particular, from \eqref{Msup=M}, this implies that we have 
\begin{equation*}
\psi_{\nu}(R)= \E\left[\frac{d\nu}{d\mb{P}}R\right]\leq 0 ,\, \text{ for all } \nu\in\mc{M}^{\mathrm{sup}}.
\end{equation*}
We may apply \cite{FK97} Lemma A.1 to show that the process $V^R$ defined by  
\begin{equation*} 
V^R_t:=\esssup_{\nu\in\mc{M}^{\mathrm{sup}}}\frac{\E\left[\frac{d\nu}{d\mb{P}}R\left|\right.\mc{F}_t\right]}
{\E\left[\frac{d\nu}{d\mb{P}}\left|\right.\mc{F}_t\right]}\quad 
\end{equation*}
with 
\begin{equation*}
V^R_0=\sup_{\nu\in\mc{M}^{\mathrm{sup}}}\E\left[\frac{d\nu}{d\mb{P}}R\right]
=\sup_{\nu\in\mc{M}^{\mathrm{sup}}}\psi_{\nu}(R)\leq 0
\end{equation*}
is a $\nu$-supermartingale for all $\nu\in\mc{M}^{\mathrm{sup}}$. Note that
here we have used the assumption that $\mc{F}_0$ is trivial. 

Using Theorem \ref{FKthm4.1} we get the existence of an $H^R\in\mc{H}$ and a nonnegative optional process $C^R$ such that $V^R=V^R_0+H^R\cdot S-C^R$.
Recall that $V_0^R\leq0$ and so at $T$, 
\begin{equation*}
(H^R\cdot S)_T\geq V^R_0+(H^R\cdot S)_T\geq V^R_0+(H^R\cdot S)_T-C^R_T=R,
\end{equation*}
that is, $R\in\mc{R}$.
The proof of the lemma is now complete.
\end{proof}

The dual problem in the present setting is the following
\begin{equation*}
w(x):=\inf_{\nu\in\mc{M}}\left(\E\left[\ti{U}\left(\frac{d\nu_c}{d\mb{P}}\right)\right]-\psi_{\nu}(B)+x\nu(\Omega)\right).
\end{equation*}
We have omitted the derivation of the dual problem. For an excellent overview
of how one should proceed given a general primal problem we refer to Rogers \cite{Rog03}. For more details on the present situation see \cite{CSW01}.


\section{Main Result and its Proof}
\label{chp_Uhalfline_proof}
Having collected all the necessary preliminaries we may state our main result. \begin{thm}
\label{chp_Uhalfline_mainthm}
Suppose that Assumptions \ref{chp_Uhalfline_Kassump}, \ref{chp_Uhalfline_UProps}, \ref{chp_Uhalfline_AEofU}, \ref{chp_Uhalfline_FGE} and \ref{ELMM} hold. For
$x>\sup_{\nu\in\mc{M}^{\mathrm{sup}}}\psi_{\nu}(B)$,
\begin{enumerate}
\item u(x)=w(x).
\item There exists $\nu^*\in\mc{M}$ optimal for $w(x)$, i.e.,
\begin{equation*}
w(x)=\E\left[\ti{U}\left(\frac{d\nu^*_c}{d\mb{P}}\right)\right]-\psi_{\nu^*}(B)+x\nu^*(\Omega).
\end{equation*}
\item There exists $H^*\in\mc{H}$  optimal for $u(x)$, i.e.,
\begin{equation*}
u(x)=\E\big[U\big(X^*-B\big)\big],
\end{equation*}
where $X^*=x+(H^*\cdot S)_T$ is the optimal terminal wealth.
\item The following relations hold, 
$$
\psi_{\nu^*_f}\big(X^*-B\big)=0,
\quad \psi_{\nu^*}\big(X^*\big)=x\nu^*(\Omega),
\quad X^*-B\in-\partial\ti{U}\left(\frac{d\nu^*_c}{d\mb{P}}\right).
$$
\end{enumerate} 
\end{thm}
\begin{rmk}
When there is no random endowment ($B\equiv0$) 
(iv) implies that $\psi_{\nu^*_f}\big(X^*\big)=0$
and $\psi_{\nu^*_c}\big(X^*\big)=x\nu^*(\Omega)$. 
In particular we could formulate our dual problem in 
$L^1(\mb{P})$ and omit the singular measure $\nu_f$, as in \cite{BTZ04,KS99,WZ08}.  

When $\nu^*\ne0$ (a sufficient condition for this is  $U(\infty)=\infty$)
we may normalize $\nu^*$ to get an equivalent probability
measure (with regular and singular parts). If we denote 
$y^*=\nu^*(\Omega)>0$ we can express parts (ii) and (iv) above
in the standard form  as those in the utility maximization literature.

Westray and Zheng \cite{WZ09} show that when $B\equiv0$ the conditions on
budget equality, subdifferential relation and feasibility are the minimal sufficient conditions for $X^*$ being a primal optimizer if the utility function $U$ is not strictly concave. In this sense, Theorem~\ref{chp_Uhalfline_mainthm} is almost a necessary and sufficient optimality condition if a dual optimizer is known to exist.   
\end{rmk}

We prove our result in three steps. First we apply a version of the Fenchel
duality theorem, Theorem \ref{chp_Uhalfline_Rockthm1}, to show the existence of a dual solution and the equality of the value functions. We then adapt a technique from \cite{KS03} to find a primal optimizer. The proof is concluded by using convex analysis to show
that the three equalities of Theorem \ref{chp_Uhalfline_mainthm} (iv) hold.
For ease of exposition each of the three steps is broken up into a series
of lemmata.
\subsection*{Step I - Equality of the Value Functions and Existence of an
Optimal Dual Solution}
We begin by showing that in the primal problem it is sufficient to take the
maximum over $\mc{C}$. Recall that $\mc{C}=\mc{R}\cap L^\infty(\mb{P})$ and
$\mc{R}_0(x)$ is the subset of $\mc{R}$ for which the expectation in (\ref{primalproblem})
is well defined. Note that if $G\in\mc{C}$ then it is bounded
and $\E\big[U(x+G-B)^+\big]$ is always finite.
We therefore avoid the problems related to the restriction from $\mc{R}$
to $\mc{R}_0(x)$, see Remark \ref{chp_Uhalfline_expfinite}. 
\begin{lem} 
\label{supK0=supC}
For all $x\in\mb{R}$
\begin{equation*}
u(x)=\sup_{R\in\mc{R}_0(x)}\E\big[U(x+R-B)\big]=
\sup_{G\in\mc{C}}\E\big[U(x+G-B)\big].
\end{equation*}
\end{lem}
\begin{proof}
Fix $x\in\mb{R}$ and observe that since $\mc{C}\subset\mc{R}_0(x)$ it is only necessary to show that
\begin{equation*}
\sup_{R\in\mc{R}_0(x)}\E\big[U(x+R-B)\big]\leq\sup_{G\in\mc{C}}\E\big[U(x+G-B)\big].
\end{equation*}
We may suppose in addition that 
\begin{equation*}
\sup_{R\in\mc{R}_0(x)}\E\big[U(x+R-B)\big]>-\infty,
\end{equation*}
otherwise the inequality is immediate.
For each $\eps>0$ we can find $H^\eps\in\mc{H}$ such that 
\begin{equation*}
\E\big[U(x+(H^\eps\cdot S)_T-B)\big]\geq\sup_{R\in\mc{R}_0(x)}\E\big[U(x+R-B)\big]-\eps.
\end{equation*}
Since $H^\eps\in\mc{H}$ we have that $(H^\eps\cdot S)_T\wedge n\in\mc{C}$ for all $n\in\mb{N}$.
If we show that there exists $n_0$ such that 
\begin{equation*}
\E\big[U(x+(H^\eps\cdot S)_T\wedge n_0-B)\big]>-\infty,
\end{equation*}
then we may apply the monotone convergence theorem and deduce, 
\begin{eqnarray*}
\sup_{G\in\mc{C}}\E\big[U(x+G-B)\big]&\geq&
\lim_{n\to\infty}\E\big[U(x+(H^\eps\cdot S)_T\wedge n-B)\big]\\
&\geq&\sup_{R\in\mc{R}_0(x)}\E\big[U(x+R-B)\big]-\eps.
\end{eqnarray*} 
This then provides the required inequality.

To find such an $n_0$ we first observe that $\E\big[U(x+(H^\eps\cdot S)_T-B)\big]>-\infty$, which implies  
\begin{equation}
\label{lowerintegrable}
\E\big[U(x+(H^\eps\cdot S)_T-B)^-\big]<\infty.
\end{equation}
A consequence of the definition of the subgradient for the concave function
$U$ is that for two points $z_1$ and $z_2$ with $z_1<z_2$, 
\begin{equation*}
\inf \partial U(z_1)\geq \sup \partial U(z_2).
\end{equation*}
From Assumption \ref{chp_Uhalfline_UProps} we can now deduce that 
\begin{equation*}
\lim_{x\to\infty}\inf \partial U(x) = 0
\end{equation*}
and thus we may choose $n_0\in\mb{N}$ with $\sup\partial U(x+n_0-b)\leq1$.
Observe that from the subgradient inequality 
\begin{equation*}
U(x+(H^\eps\cdot S)_T-B)\leq U(x+(H^\eps\cdot S)_T\wedge n_0-B)
+q\big((H^\eps\cdot S)_T-(H^\eps\cdot S)_T\wedge n_0\big), 
\end{equation*}
for any $q\in\partial U\big(x+(H^\eps\cdot S)_T\wedge n_0-B\big)$. This continues
to hold if we multiply both sides by $\I{\{(H^\eps\cdot S)_T\geq n_0\}}$.
Since on $\{(H^\eps\cdot S)_T\geq n_0\}$ we have that any $q\in\partial U\big(x+(H^\eps\cdot S)_T\wedge n_0-B\big)$ satisfies $|q|\leq1$ we deduce that
\begin{equation*}
\begin{split}U(x+(&H^\eps\cdot S)_T\wedge n_0-B)\I{\{(H^\eps\cdot S)_T\geq n_0\}}\\
&\geq \big(U(x+(H^\eps\cdot S)_T-B)-2|(H^\eps\cdot S)_T|\big)
\I{\{(H^\eps\cdot S)_T\geq n_0\}}.
\end{split}
\end{equation*}
On the set $\{(H^\eps\cdot S)_T<n_0\}$ we have that $(H^\eps\cdot S)_T=(H^\eps\cdot S)_T\wedge n_0$. Combining the above two estimates gives
\begin{equation*}
U(x+(H^\eps\cdot S)_T\wedge n_0-B)^-\leq 2\big(U(x+(H^\eps\cdot S)_T-B)^-+
|(H^\eps\cdot S)_T|\big).  
\end{equation*}
It now follows from \eqref{lowerintegrable} together with $H^\eps\in\mc{H}$
that $U(x+(H^\eps\cdot S)_T\wedge n_0-B)^-\in L^1(\mb{P})$. It
then must be the case that 
\begin{equation*}
\E\big[U(x+(H^\eps\cdot S)_T\wedge n_0-B)\big]>-\infty.
\end{equation*}
This provides the existence of such an $n_0$ and completes the proof.
\end{proof}

We use the method of \cite{BF02,BF05} to establish Theorem \ref{chp_Uhalfline_mainthm} (i) and (ii). The key result needed in the proof is the following Fenchel
Duality Theorem, stated for the dual system $\big(L^\infty(\mb{P}),ba(\mb{P})\big)$.

\begin{thm}[Rockafellar \cite{R66} Theorem 1]
\label{chp_Uhalfline_Rockthm1}
Suppose $\alpha:L^\infty(\mb{P})\to\mb{R}\cup\{+\infty\}$ and $\beta:L^\infty(\mb{P})\to\mb{R}\cup\{-\infty\}$ are respectively proper convex and concave functionals. If either $\alpha$ or $\beta$ is continuous at some point where both functions are finite then
\begin{equation*}
\sup_{G\in L^{\infty}(\mb{P})}\{\beta(G)-\alpha(G)\}=\min_{v\in ba(\mb{P})}\{\alpha^*(\nu)-\beta^*(\nu)\},
\end{equation*}
where the functionals $\alpha^*$ and $\beta^*$ are respectively the \textit{Fenchel convex conjugate} and \textit{Fenchel concave conjugate} defined on $ba(\mb{P})$ by 
\begin{equation*}
\alpha^*(\nu):=\sup_{G\in L^\infty(\mb{P})}\{\psi_{\nu}(G)-\alpha(G)\} \text{ and }
\beta^*(\nu):=\inf_{G\in L^\infty(\mb{P})}\{\psi_{\nu}(G)-\beta(G)\}.
\end{equation*}
\end{thm}
The importance of Theorem \ref{chp_Uhalfline_Rockthm1} is that it not only
shows the equality of the primal and dual value functions but also establishes
the existence of the optimal dual solution. The following result provides
the details.

\begin{thm} 
\label{u=w}
For $x>\sup_{\nu\in\mc{M}^{\mathrm{sup}}}\psi_{\nu}(B)$ we have
\begin{equation*}
u(x)=w(x)=\min_{\nu\in\mc{M}}\left(\E\left[\ti{U}\left(\frac{d\nu_c}{d\mb{P}}\right)\right]-\psi_{\nu}(B)+x\nu(\Omega)\right).
\end{equation*}
\end{thm}
\begin{proof} Let $x_1:=\sup_{\nu\in\mc{M}^{\mathrm{sup}}}\psi_{\nu}(B)$ so that $x>x_1$
and define the concave
functional $I_{U^{x,B}}:L^\infty(\mb{P})\to\mb{R}\cup\{-\infty\}$ by 
\begin{equation*}
I_{U^{x,B}}(G):=\E\big[U(x+G-B)\big].
\end{equation*}
We write $\delta_{\mc{C}}$ for the indicator function in the sense of convex
analysis, so that  
\begin{equation*}
\delta_{\mc{C}}(G)=\left\{\begin{array}{ll}
 0 & \text{ for $G\in\mc{C}$},\\
\infty & \text{ for $G\notin\mc{C}$}. \\ \end{array}\right.
\end{equation*}
Using Lemma \ref{supK0=supC} we  can write 
\begin{equation*}
u(x)=\sup_{G\in L^\infty(\mb{P})}\{I_{U^{x,B}}(G)-\delta_{\mc{C}}(G)\}.
\end{equation*}
We first construct an $R^B$ for which $R^B\in \mc{C}$, $\big|I_{U^{x,B}}(R^B)\big|<\infty$ and $I_{U^{x,B}}$
is continuous at $R^B$. We can then apply Theorem \ref{chp_Uhalfline_Rockthm1}
to get
\begin{equation}
\label{chp_Uhalfline_eq:1}
u(x)=\min_{\nu\in ba(\mb{P})}\{\delta^*_{\mc{C}}(\nu)-I^*_{U^{x,B}}(\nu)\}.
\end{equation}
Proceeding as in Lemma \ref{InKiffnuK}
 we deduce the existence of $H^B\in\mc{H}$ such that $x_1+(H^B\cdot S)_T\geq
B$. Since  $\|B\|_{L^\infty(\mb{P})}:=b<\infty$, for $m\geq b-x_1$ we have
\begin{equation*}  
x_1+(H^B\cdot S)_T\wedge m\geq B.
\end{equation*}
Pick such an $m_0$ and write $R^B=(H^B\cdot S)_T\wedge m_0$, an element of $\mc{C}$.
We have the following inequalities for $I_{U^{x,B}}(R^B)$,
\begin{equation*} 
U(x-x_1)\leq I_{U^{x,B}}(R^B)\leq U(x+m_0+b).
\end{equation*}
This implies that 
\begin{equation*}
\big|I_{U^{x,B}}(R^B)\big|\leq \max\{|U(x+m_0+b)|,|U(x-x_1)|\}<\infty.
\end{equation*}
We now show that $I_{U^{x,B}}$ is continuous at $R^B$ with respect to the
norm topology on $L^\infty(\mb{P})$. Suppose $(G_n)_{n\in\mb{N}}$ is a sequence
in $L^\infty(\mb{P})$ converging to $R^B$ and set $\eps_0:=(x-x_1)/2$. For
$\|G_n-R^B\|_{L^\infty(\mb{P})}<\eps_0$ we have 
the estimate
\begin{equation*}
\big|U(x+G^n-B)\big|\leq \max\{|U(x+m_0+\eps_0+b)|,|U(\eps_0)|\}.
\end{equation*}
The dominated convergence theorem now implies that $I_{U^{x,B}}$ is continuous
at $R^B$. The conditions of Theorem \ref{chp_Uhalfline_Rockthm1} are now
satisfied and \eqref{chp_Uhalfline_eq:1} follows.

The next step is to prove that 
\begin{equation*}
\min_{\nu\in ba(\mb{P})}\{\delta^*_{\mc{C}}(\nu)-I^*_{U^{x,B}}(\nu)\}
=\min_{\nu\in\mc{M}}\left(\E\left[\ti{U}\left(\frac{d\nu_c}{d\mb{P}}\right)\right]-\psi_{\nu}(B)+x\nu(\Omega)\right).
\end{equation*}
Using the fact that $\mc{C}$ is a cone one can show 
\begin{equation}
\label{u=w:1}
 \delta^*_{\mc{C}}(\nu)=\sup_{G\in \mc{C}}\{\psi_{\nu}(G)\}=\delta_{(\mc{C})^0}(\nu).
\end{equation} 
Here $(\mc{C})^0$ denotes the polar of the cone $\mc{C}$.
In addition we have
\begin{eqnarray}
\nonumber I^*_{U^{x,B}}(\nu)&=&\inf_{G\in L^\infty(\mb{P})}\{\psi_{\nu}(G)-I_{U^{x,B}}(G)\}\\
\label{u=w:2} &=&I^*_U(\nu)+\psi_{\nu}(B)-x\nu(\Omega),
\end{eqnarray}
where we have performed the change of variables $F:=x+G-B$ and defined 
\begin{equation*}
I_U(F):=\E[U(F)]. 
\end{equation*}
Now $I_U$ is a normal concave integrand
in the sense of \cite{R71}. To characterize
$I^*_U$ we use \cite{BF02} Lemma 3.1, derived from \cite{R71} Theorems 1 and 2. It states that 
\begin{equation}
\label{u=w:3}
I^*_U(\nu)=I_{U^*}(\nu_c)-\delta^*_{\D(I_U)}(-\nu_f) \text{
for } \nu\in ba(\mb{P}),
\end{equation}
where $I_{U^*}(\nu_c):=\E[U^*(d\nu_c/d\mb{P})]$ and 
\begin{equation}
\label{u=w:4}
U^*(y):=\inf_{x>0}\{xy-U(x)\}=-\sup_{x>0}\{U(x)-xy\}=-\ti{U}(y).
\end{equation}
Combining (\ref{u=w:1}), (\ref{u=w:2}), (\ref{u=w:3}) and (\ref{u=w:4}) we see
that 
\begin{equation*}
\begin{split}
\min_{\nu\in ba(\mb{P})}&\{\delta^*_{\mc{C}}(\nu)-I^*_{U^{x,B}}(\nu)\}\\
=&\min_{\nu\in ba(\mb{P})\cap(\mc{C})^0}
\left(\E\left[\ti{U}\left(\frac{d\nu_c}{d\mb{P}}\right)\right]-\psi_{\nu}(B)+x\nu(\Omega)+\delta^*_{\D(I_U)}(-\nu_f)\right).
\end{split}
\end{equation*}
Suppose $\nu\in\mc{M}=(\mc{C})^0$, we want to show that $\delta^*_{\D(I_U)}(-\nu_f)=0$. Indeed, as $L^{\infty}_{-}(\mb{P})\subseteq \mc{C}$ we know that $\nu\geq0$. In particular this implies that $\nu_f\geq0$ and so $\psi_{\nu_f}(G)\geq0$ for all $G\in L^{\infty}_{+}(\mb{P})$. Observe that $U(x)=-\infty$ for $x<0$
and so it must be the case that $\D(I_U)\subseteq L^{\infty}_{+}(\mb{P})$. The above discussion allows us to conclude
\begin{equation*}
\psi_{\nu_f}(G)\geq 0 \text{ for all } G\in\D(I_U) \text{ and } \nu\in\mc{M}.
\end{equation*}
As we have $-\psi_{\nu_f}(G)=\psi_{-\nu_f}(G)$ we see that 
\begin{equation*}
\delta^*_{\D(I_U)}(-\nu_f) = \sup_{G\in\D(I_U)}\{\psi_{-\nu_f}(G)\}\leq 0.
\end{equation*}
However, for all $\eps>0$ we have that $\eps\in\D(U)$ and hence for all $\nu\in\mc{M}$,
\begin{equation*}
-\eps\nu_f(\Omega) \leq \delta^*_{\D(I_U)}(-\nu_f) \leq 0.
\end{equation*}
We then conclude that $\delta^*_{\D(I_U)}(-\nu_f)=0$ for all $\nu\in\mc{M}$.
Thus we can write \eqref{chp_Uhalfline_eq:1} as    
\begin{eqnarray*}
u(x)=\min_{\nu\in\mc{M}}
\left(\E\left[\ti{U}\left(\frac{d\nu_c}{d\mb{P}}\right)\right]-\psi_{\nu}(B)+x\nu(\Omega)\right).
\end{eqnarray*}
This completes the proof of items (i) and (ii). 
\end{proof}
\begin{rmk}
\label{BF02rmk} In \cite{BF02} the authors apply a version of the Fenchel
duality theorem from Luenberger \cite{L69} for which it is necessary that the set 
$\D(I_U)\cap\D(\delta_{\mc{C}})$ contains an interior point.
This is nontrivial to check and we therefore choose to use an alternative version of the Fenchel duality theorem, Theorem \ref{chp_Uhalfline_Rockthm1}, and prove the existence of a continuity point directly.  
\end{rmk}

\subsection*{Step II - Existence for the Primal Problem}
Key in the proof of the existence of a primal optimizer and corresponding replicating strategy is the ``dual'' representation of $\mc{R}$ given in Lemma  \ref{InKiffnuK}, the following is the crucial
result.
\begin{lem} 
\label{u=U(R)}
For each $x>\sup_{\nu\in\mc{M}^{\mathrm{sup}}}\psi_{\nu}(B)$ there exists $H^*\in\mc{H}$ such that 
\begin{equation*}
u(x)=\E\big[U\big(x+(H^*\cdot S)_T-B\big)\big]. 
\end{equation*}
\end{lem}
\begin{proof}
We know from Lemma \ref{u=w} that $|u(x)|<\infty$ for $x>\sup_{\nu\in\mc{M}^{\mathrm{sup}}}\psi_{\nu}(B)$.
In addition since $B\geq-b$ and $\mc{M}^{\mathrm{sup}}\subset ba_+(\mb{P})$ we have 
\begin{equation*}
x>\sup_{\nu\in\mc{M}^{\mathrm{sup}}}\psi_{\nu}(B)\geq-b,
\end{equation*}
so that $x+b>0$. Now fix $x$ and take a sequence $(R_n)_{n\in\mb{N}}$ with each $R^n\in\mc{R}$ such that
\begin{equation*}
\lim_{n\to\infty}\E[U(x+R_n-B)]=u(x).
\end{equation*}
Since $U(x)=-\infty$ for $x<0$, by passing to a subsequence if necessary, we may assume that 
\begin{equation*}
x+R_n-B\geq 0 \text{ a.s. for all $n$}.
\end{equation*}
In particular the sequence $(R_n)_{n\in\mb{N}}$ is bounded below by a constant,
uniformly in $n$. We may now apply Delbaen and Schachermayer
\cite{DS94} Lemma A1.1 to find, for
each $n\in\mb{N}$, a sequence of convex combinations $(\lambda_{n,m})_{m\geq
n}$ and a random variable $R_*$ such that
\begin{equation*}
R_n^1:=\sum_{m\geq n}\lambda_{n,m}R_m \xrightarrow{} R_* \text{ a.s. }
\end{equation*}
Each $R^1_n\in\mc{R}$ and the
sequence $(R^1_n)_{n\in\mb{N}}$ is bounded from below, uniformly in $n$, and thus, for $\nu\in\mc{M}^c$, 
\begin{equation*}
\psi_{\nu}(R^1_n)=\E\left[\frac{d\nu}{d\mb{P}}R^1_n\right]\leq 0.
\end{equation*}
Applying Fatou's lemma we see that 
\begin{equation*}
\psi_{\nu}(R_*)=\E\left[\frac{d\nu}{d\mb{P}}R_*\right]\leq 0 \, \text{ for all
} \nu\in\mc{M}^c.
\end{equation*}
From Lemma \ref{InKiffnuK} we deduce that $R_*\in\mc{R}$.
Using the concavity of $U$ we have the following,
\begin{equation*}
u(x)\geq \E[U(x+R^1_n-B)]\geq\sum_{m\geq n}\lambda_{n,m}\E[U(x+R_m-B)].
\end{equation*}
This implies that $\big(\E[U(x+R^1_n-B)]\big)_{n\in\mb{N}}$ also converges
to $u(x)$. Exactly as in \cite{KS03} Lemma 1 if we show that 
$\big(U(x+R^1_n-B)^+\big)_{n\in\mb{N}}$ is uniformly integrable the proof will
be complete. For then, by reverse Fatou's lemma and noting $R_*\in\mc{R}$, \begin{equation*}
u(x)=\limsup_{n\to\infty}\E[U(x+R^1_n-B)]\leq \E[U(x+R_*-B)]\leq u(x).
\end{equation*}
As $R_*\in\mc{R}$ we know there exists some $H^*\in\mc{H}$ with 
$(H^*\cdot S)_T\geq R_*$. Since $U$ is increasing this provides 
\begin{equation*}
u(x)=\E[U(x+R_*-B)]\leq\E\big[U\big(x+(H^*\cdot S)_T-B\big)\big] \leq u(x).
\end{equation*}
The statement of the lemma then follows.

Thus we suppose for a contradiction that the uniform integrability fails. Exactly as in \cite{KS03} we may find a sequence of disjoint sets $(A_n)_{n\in\mb{N}}$ contained in $\mc{F}$ and an $\eps>0$ such that, after possibly passing to a subsequence, again indexed by $n$,
\begin{equation}
\label{u=U(R):3}
\E\left[U\big(x+R^1_n-B\big)^+\I{A_n}\right]\geq\eps.
\end{equation}
Using the $\nu^1$ from Assumption \ref{chp_Uhalfline_FGE}, together with \eqref{dualfinite} and Lemma \ref{u=w}
we see that for all $r>0$ and $z>\sup_{\nu\in\mc{M}^{\mathrm{sup}}}\psi_{\nu}(B)$,
\begin{equation}
\label{u=U(R):1}
u(z)=w(z)\leq \E\left[\ti{U}\left(r\frac{d\nu^1_c}{d\mb{P}}\right)+r(b+z)\frac{d\nu^1_c}{d\mb{P}}\right]<\infty.
\end{equation}
For $z$ sufficiently large $u(z)\geq U(z-b)\geq0$ so that combining this
with \eqref{u=U(R):1} we deduce, 
\begin{equation*}
0\geq \limsup_{z\to\infty} \frac{u(z)}{z}\geq\liminf_{z\to\infty}\frac{u(z)}{z}\geq0.
\end{equation*} 
If $U(z)\leq0$ for all $z>0$ then $U^+$ is identically $0$ and the uniform integrability is immediate, hence we may assume there exists $z>0$ such that
$U(z)>0$. 

Define $x_2<\infty$ by  
\begin{equation*}
x_2:=\inf\{z\geq x\,:\,U\big(z-2(b+x)\big)>0\},
\end{equation*}
as well as the sequence $(R^2_n)_{n\in\mb{N}}$ via 
\begin{equation*}
R^2_n:=\sum_{m=1}^{n}R^1_m \I{A_m}.
\end{equation*}
As each $R^1_m\geq-(x+b)$ and $x+b>0$ we have $R^2_n\geq -(x+b)$ for all
$n\in\mb{N}$ and  
\begin{equation*}
R^2_n\leq \sum_{m=1}^{n}\big(R^1_m+(x+b)\big)\I{A_m}\leq \sum_{m=1}^{n}R^1_m+n(x+b).
\end{equation*}
Let $\nu\in\mc{M}^c$, from the above we see that
\begin{equation*}
\psi_{\nu}\big(R^2_n-n(x+b)\big)\leq
\sum_{m=1}^{n}\E\left[\frac{d\nu}{d\mb{P}}R^1_m\right]\leq 0. 
\end{equation*}
Using Lemma \ref{InKiffnuK} we see that for each $n\in\mb{N}$, $R^2_n-n(x+b)\in\mc{R}$.
We claim that, in addition, for each $n\in\mb{N}$,
\begin{equation}
\label{u=U(R):2}
U(x_2+R_n^2-B)\geq\sum_{m=1}^{n}U\big(x+R^1_m-B\big)^+\I{A_m}. 
\end{equation}
Indeed, let us fix $n\in\mb{N}$. If $\omega\notin A_m$ for all $m$ then
the right hand side is $0$ and the left hand side satisfies  
\begin{equation*}
U(x_2+R_n^2-B)\geq U\big(x_2-(x+b)-b\big)\geq 0.
\end{equation*}
If $\omega\in\bigcup_{m=1}^n A_m$ then since the sequence of sets $(A_n)_{n\in\mb{N}}$ is mutually disjoint $\omega\in
A_{m_0}$ for some unique $m_0$ satisfying $1\leq m_0 \leq n$. As $x_2\geq x$ by construction and $U$ is increasing,  
\begin{eqnarray*}
U(x_2+R_n^2-B)=U\big(x_2+R^1_{m_0}-B\big)\geq U\big(x+R^1_{m_0}-B\big).
\end{eqnarray*}
Since $R^1_n\geq-(x+b)$ for all $n$ we see from the definition of $x_2$
that $U\big(x_2+R^1_{m_0}-B\big)\geq 0$. It follows that (\ref{u=U(R):2})
holds.
Recalling that $R^2_n-n(x+b)\in\mc{R}$ as well as using (\ref{u=U(R):3}) and (\ref{u=U(R):2}) we have that
\begin{eqnarray*}
\limsup_{z\to\infty}\frac{u(z)}{z}
&\geq &\limsup_{n\to\infty}\frac{\E\left[U\big(x_2+n(x+b)+R^2_n-n(x+b)-B\big)\right]}{x_2+n(x+b)}\\
&\geq&\limsup_{n\to\infty}\frac{n\eps}{x_2+n(x+b)}=\frac{\eps}{x+b}>0.
\end{eqnarray*}
This is our contradiction and completes the proof of item (iii).
\end{proof}

\subsection*{Step III - The Duality Relations}
\begin{lem}
Let $\nu^*$ be the optimal dual solution 
and $X^*=x+(H^*\cdot S)_T$ be the optimal terminal wealth
as in Lemma
\ref{u=U(R)} then we have
$$
\psi_{\nu^*_f}\big(X^*-B\big)=0,
\quad\psi_{\nu^*}\big(X^*\big)=x\nu^*(\Omega), \quad
X^*-B\in-\partial\ti{U}\left(\frac{d\nu^*_c}{d\mb{P}}\right).
$$
\end{lem}  
\begin{proof}
Since $\D(U)\subset \mb{R}_+$ and $\left|\E\big[U\big(X^*-B\big)\big]\right|<\infty$ we
have $X^*-B\geq 0$. This implies that $\psi_{\nu^*_f}\big(X^*-B\big)\geq0$. 
Now suppose
for a contradiction that this is strict. 
As $\big((H^*\cdot S)_T\big)^-\in L^\infty(\mb{P})$ we have that $(H^*\cdot S)_T\wedge n$ is
in $\mc{C}$ for all $n$. Since $\nu^*\in\mc{M}$ this implies that $\psi_{\nu^*}\big((H^*\cdot S)_T\wedge n\big)\leq0$  and hence $\psi_{\nu^*}\big((H^*\cdot S)_T\big)\leq 0$. Using this together with the conjugate
relations we have,
\begin{eqnarray*}
u(x)&=&\E\big[U\big(X^*-B\big)\big]\\
&\leq&\E\left[\ti{U}\left(\frac{d\nu^*_c}{d\mb{P}}\right)\right]+
\psi_{\nu_c^*}\big(X^*-B\big)\\
&<&\E\left[\ti{U}\left(\frac{d\nu^*_c}{d\mb{P}}\right)\right]-\psi_{\nu^*}(B)+x\nu^*(\Omega)=w(x).
\end{eqnarray*} 
Since we have proved that $u(x)=w(x)$ we have our contradiction.

We know that $\psi_{\nu^*}\big((H^*\cdot S)_T\big)\leq 0$. To show that $\psi_{\nu^*}\big((H^*\cdot S)_T\big)=0$ assume for a contradiction that the inequality is strict. We now have, using $\psi_{\nu^*_f}\big(X^*-B\big)=0$, 
\begin{eqnarray*}
u(x)&\leq&\E\left[\ti{U}\left(\frac{d\nu^*_c}{d\mb{P}}\right)\right]
+\psi_{\nu_c^*}\big(X^*-B\big)\\
&<&\E\left[\ti{U}\left(\frac{d\nu^*_c}{d\mb{P}}\right)\right]-\psi_{\nu^*}(B)+x\nu^*(\Omega)=w(x).
\end{eqnarray*} 
This is again a contradiction. To establish that $X^*-B$ is in the appropriate
subgradient, we use \cite{R70} Theorem 23.5. This states that
for conjugate functions $U$ and $\ti{U}$,
\begin{gather*}
U(x)\leq\ti{U}(y)+xy \text{ for all }y\geq0,  \\
U(x)=\ti{U}(y)+xy \text{ if and only if } x\in-\partial \ti{U}(y).
\end{gather*}
Assume for a contradiction that the set
\begin{equation*}
\Lambda:=\left\{\omega\in\Omega:X^*(\omega)-B(\omega)\notin-\partial \ti{U}\left(\frac{d\nu^*_c(\omega)}
{d\mb{P}}\right)\right\},
\end{equation*}
satisfies $\mb{P}(\Lambda)>0$.
We now have, using $\psi_{\nu^*_f}\big(X^*-B\big)=0$ and $\psi_{\nu^*}\big((H^*\cdot S)_T\big)$=0, 
\begin{eqnarray*}
u(x)&<&\E\left[\ti{U}\left(\frac{d\nu^*_c}{d\mb{P}}\right)\right]
+\psi_{\nu_c^*}\big(X^*-B\big)\\
&=&\E\left[\ti{U}\left(\frac{d\nu^*_c}{d\mb{P}}\right)\right]-\psi_{\nu^*}(B)+x\nu^*(\Omega)=w(x).
\end{eqnarray*}
This is again a contradiction and the proof of the item (iv) is complete.
\end{proof}
\appendix
\bibliography{Max_on_R+}
\bibliographystyle{abbrv}

\end{document}